\documentclass[11pt]{article}
\usepackage{amsmath,amssymb,bbm,amsthm}
\usepackage{fullpage}
\usepackage{thm-restate,color,xcolor,xspace}
\usepackage{hyperref,cleveref}
\usepackage{graphicx}
\usepackage{algorithm,algorithmic}

\newtheorem{theorem}{Theorem}
\newtheorem{lemma}[theorem]{Lemma}
\newtheorem{definition}[theorem]{Definition}

\begin{document}

\newcommand{\bu}{\mathbf u}
\newcommand{\bg}{\mathbf g}
\newcommand{\by}{\mathbf y}
\newcommand{\bc}{\mathbf c}

\title{Maximum Flow Without the Outer IPM}
\author{Jason Li\footnote{Carnegie Mellon University. \tt jmli@cs.cmu.edu} \and
Alex Wice\footnote{\tt alexwice@gmail.com}}
\date{\today}
\maketitle

\begin{abstract}
We show that the balancing weights technique of Li (2026) actually produces an approximate \emph{pseudo-circulation} of a directed, capacitated graph in $m^{1+o(1)}$ time. Together with standard flow techniques, we obtain an $m^{1+o(1)}$ time maximum flow algorithm that avoids the interior-point method framework of recent almost-linear time algorithms (Chen et al.\ FOCS 2022, van den Brand et al.\ FOCS 2024).
\end{abstract}

\section{Introduction}

Recent breakthroughs have obtained the first almost-linear time algorithms for maximum flow \cite{chen2025maximum,van2024almost}. The algorithms have two main components: an outer \emph{interior-point method} (IPM) that reduces maximum flow to $m\,\textup{polylog}\,n$ specialized subproblems, and an inner dynamic data structure that efficiently solves these subproblems in $n^{o(1)}$ amortized time. In this note, we show that the dynamic \emph{min-ratio cut} data structure from~\cite{van2024almost} can be adapted to compute a \emph{pseudo-circulation} of the input graph. The pseudo-circulation can be transformed to an approximate maximum flow by standard techniques, and iterating on the residual graph obtains an $m^{1+o(1)}$ time algorithm that avoids the outer IPM loop common to all previous approaches.

We begin with some preliminaries. For a directed graph, the out-boundary $\partial^+S$ is the set of arcs with only the tail in $S$, and the in-boundary $\partial^-S$ is the set of arcs with only the head in $S$. For a capacity function $\bc:E\to\mathbb R_{\ge0}$ on the arcs, let $\bc(\partial^+S)$ denote the total capacity of arcs in $\partial^+S$, and define $\bc(\partial^-S)$ similarly.

Given a capacitated graph $G=(V,E,\bc)$, we represent a flow as a function $f:E\to\mathbb R_{\ge0}$ satisfying $f(u,v)\le c(u,v)$ for all arcs $(u,v)\in E$. For a subset of arcs $S\subseteq E$, let $f(S)=\sum_{(u,v)\in S}f(u,v)$ denote the total flow sent along arcs in $S$. Given distinct vertices $s,t\in V$, an \emph{$s$--$t$ flow} is a flow that also satisfies flow conservation constraints, namely $f(\partial^+\{v\})=f(\partial^-\{v\})$ for all $v\in V\setminus\{s,t\}$. An \emph{$s$--$t$ cut} is a set $S\subseteq V$ satisfying $s\in S$ and $t\notin S$.

We now define a $\beta$-pseudo-circulation, which is exactly a standard circulation when $\beta=1$.

\begin{definition}[$\beta$-pseudo-circulation]
Given a parameter $\beta\ge1$, a \emph{$\beta$-pseudo-circulation} is a flow $f$ satisfying $f(\partial^+S)\le\beta\cdot f(\partial^-S)$ for all $S\subseteq V$. (Observe that this requirement is actually two-sided, i.e., $f(\partial^-S)\le\beta\cdot f(\partial^+S)$ follows from replacing $S$ by $V\setminus S$.)
\end{definition}

For a pseudo-circulation to approximately capture the maximum $s$--$t$ flow, the arc $(t,s)$ must exist and carry sufficient flow.

\begin{definition}
Given two vertices $s,t\in V$ and a capacitated graph that includes the arc $(t,s)$, the \emph{$s$--$t$ value} of a pseudo-circulation $f$ is $f(t,s)$.
\end{definition}

Our main technical contribution is the following.

\begin{theorem}\label{thm:main}
Consider a capacitated graph $G=(V,E,\bc)$ and distinct vertices $s,t\in V$ such that all arc capacities are integral and polynomially bounded, except for arc $(t,s)$ which has infinite capacity. Let $F$ be the value of the maximum $s$--$t$ flow. There is a deterministic $m^{1+o(1)}$ time algorithm that computes a $2$-pseudo-circulation with $s$--$t$ value at least $F/6$.
\end{theorem}

\section{From pseudo-circulation to maximum flow}

We first show how to convert a pseudo-circulation into an $s$--$t$ flow of approximately the same value. The main observation is that a pseudo-circulation induces a capacitated subgraph that is \emph{balanced}.

\begin{definition}[Balance]
A capacitated, strongly connected directed graph $G=(V,E,\bc)$ is \emph{$\beta$-balanced} if $\bc(\partial^+S)\le\beta\cdot\bc(\partial^-S)$ for all $S\subseteq V$.
\end{definition}

Consider the $2$-pseudo-circulation with $s$--$t$ value at least $F/6$ guaranteed by \Cref{thm:main}. It induces a capacitated graph $G|_f=(V,E,f)$ that is $2$-balanced by definition. We now claim that $G|_f$ supports a constant fraction of the maximum $s$--$t$ flow from $G$.

\begin{lemma}
The maximum $s$--$t$ flow in $G|_f$ is at least $F/12$.
\end{lemma}
\begin{proof}
By the max-flow/min-cut theorem, it suffices to show that any $s$--$t$ cut $S\subseteq V$ satisfies $f(\partial^+S)\ge F/12$. Since $f$ is a $2$-pseudo-circulation, we have $f(\partial^+S)\ge f(\partial^-S)/2$. Finally, $f(\partial^-S)\ge F/6$ since $(t,s)\in\partial^-S$ and the pseudo-circulation has value $f(t,s)\ge F/6$.
\end{proof}

It is known that balanced graphs behave similarly to undirected graphs and admit approximate maximum flow algorithms:

\begin{theorem}[Theorem 2.8 of~\cite{ene2016routing}]
Given a $\beta$-balanced directed graph and a parameter $\epsilon>0$, there is an algorithm that finds a $(1-\epsilon)$-approximate $s$-$t$ max-flow in randomized $\tilde{O}(m\beta^2/\epsilon^2)$ time.
\end{theorem}

Setting $\epsilon=1/2$, the algorithm computes a half-approximate maximum $s$--$t$ flow on $G|_f$, which has value at least $F/24$. The same $s$--$t$ flow is therefore a constant-approximate maximum flow on the original graph $G$. Round the flow to an integral flow, update the residual graph, and iterate this process $O(\log m)$ times to obtain an exact maximum flow on a graph with integral and polynomially-bounded arc capacities. For integral arc capacities in the range $\{1,2,\ldots,U\}$, standard reductions incur an additional $\log U$ factor.

\section{Pseudo-circulation implies approximation}

In this section, we construct a pseudo-circulation that establishes \Cref{thm:main}. Similar to~\cite{li2026balancing}, the pseudo-circulation is determined by a potential function $\by:V\to\mathbb R$ on the vertices. For a given potential function $\by$, the corresponding flow on arc $(u,v)$ is
\[ f(u,v)=\frac1{\max\{\by(v)-\by(u),\,0\}+1/c(u,v)} .\]
In particular, we have $f(u,v)=\frac1{\by(v)-\by(u)+1/c(u,v)}$ whenever $\by(v)\ge\by(u)$, and $f(u,v)=c(u,v)$ whenever $\by(v)\le\by(u)$. Viewing the potential function as heights on the vertices, observe that the flow must send full capacity along ``downward'' arcs $(u,v)$ with $\by(v)\le\by(u)$. We can interpret this flow as a continuous version of the push-relabel algorithm, where sufficiently ``downward'' arcs also send full flow. Compared to push-relabel, our flow has the advantage of being completely determined by the potential function. On the other hand, our flow is not a proper $s$--$t$ flow or even a pre-flow, so instead we aim for a pseudo-circulation.

In the next section, we obtain an algorithm that constructs a potential function that induces a pseudo-circulation.

\begin{restatable}[Finding a pseudo-circulation]{lemma}{Finding}\label{lem:finding}
Consider a strongly connected, capacitated graph $G=(V,E,\bc)$ with integral and polynomially-bounded arc capacities, and consider distinct vertices $s,t\in V$. For any parameter $\epsilon>0$, there is a deterministic $m^{1+o(1)}/\epsilon^3$ time algorithm that computes a potential function $\by$ whose corresponding flow is a $(1+\epsilon)$-pseudo-circulation.
\end{restatable}

Surprisingly, finding a potential function that induces a pseudo-circulation is the only difficult task: we show that if the graph is \emph{preprocessed} beforehand, then this pseudo-circulation must have large $s$--$t$ value, fulfilling the requirement of \Cref{thm:main}.

We first describe the preprocessing, which is similar to Section~1.4 of~\cite{li2026balancing}. First, compute strongly connected components of the graph and discard any strongly connected component not containing $s$ or $t$. If $s$ and $t$ do not belong in the same strongly connected component, then the existence of arc $(t,s)$ means that there is no path from $s$ to $t$, so the maximum flow value is $0$ and the algorithm can return the empty circulation. Otherwise, the remaining strongly connected component contains $s$ and $t$, and it is not hard to see that the maximum $s$--$t$ flow on this component is still $F$.

Next, assume that all finite arc capacities are in the range $\{1,2,\ldots,U\}$ for some $U=n^{O(1)}$. We replace the infinite-capacity arc $(t,s)$ by many parallel arcs $(t,s)$ of capacity $U$ each, called \emph{anchors}, such that for the new number of arcs $m$, there are exactly $m/2$ anchors. The value of a pseudo-circulation is now defined as the total flow sent along the anchors $(t,s)$.

\begin{lemma}[Pseudo-circulation implies approximation]
Consider distinct vertices $s,t\in V$ and a preprocessed graph $G=(V,E,\bc)$, and consider a potential function $\by$ that induces a $2$-pseudo-circulation $f$. Then, the $s$--$t$ value of $f$ is at least $F/6$.
\end{lemma}
\begin{proof}
If $\by(s)-\by(t)<1/U$, then each of the $m/2$ anchors $(t,s)$ sends at least $U/2$ flow, so the $s$--$t$ value is at least $m/2\cdot U/2=mU/4\ge F/4$. For the rest of the proof, assume that $\by(s)-\by(t)\ge1/U$.

For each $\lambda\in(\by(t),\by(s))$, define $S_\lambda=\{v\in V:\by(v)\ge\lambda\}$, which is an $s$--$t$ cut. By the max-flow/min-cut theorem, the arcs in $\partial^+S_\lambda$ have total capacity at least $F$. All of these arcs $(u,v)$ go ``downhill'', i.e., $\by(u)>\by(v)$, so the flow sends full capacity along each arc, i.e., $f(u,v)=c(u,v)$. It follows that $f(\partial^+S_\lambda)\ge F$ for all $\lambda\in(\by(t),\by(s))$. Since $f$ is a $2$-pseudo-circulation, we also have $f(\partial^-S_\lambda)\ge F/2$.

We now consider the integral
\[ I=\int_{\by(t)}^{\by(s)} f(\partial^-S_\lambda)\,d\lambda. \]
Since $f(\partial^-S_\lambda)\ge F/2$ for all $\lambda\in(\by(t),\by(s))$, we have
\[ I\ge(\by(s)-\by(t))\cdot F/2 .\]
Now consider the individual contribution of each arc to the integral. Let $A=m/2\cdot\frac1{\by(s)-\by(t)+1/U}$ denote the total flow sent along the anchors. Each anchor $(t,s)$ is always in $f(\partial^-S_\lambda)$, so the total contribution of anchors to $I$ is
\[ A\cdot(\by(s)-\by(t))=m/2\cdot\frac{\by(s)-\by(t)}{\by(s)-\by(t)+1/U}\ge m/4 ,\]
where the inequality uses $\by(s)-\by(t)\ge1/U$.
For every non-anchor $(u,v)$ with $\by(u)<\by(v)$, it contributes exactly $1/(\by(v)-\by(u)+1/c(u,v))$ to $f(\partial^-S_\lambda)$ for all $\lambda$ in an interval of length $\by(v)-\by(u)$, so its overall contribution is at most $1$. It follows that the $m/2$ non-anchors contribute at most $m/2$ to $I$.

We conclude that the anchors contribute at least $1/3$ fraction to the integral $I$. It follows that
\[ A\cdot(\by(s)-\by(t))\ge I/3\ge(\by(s)-\by(t))\cdot F/6 ,\]
so $A\ge F/6$, concluding the proof.
\end{proof}

\section{Finding a pseudo-circulation}

For the rest of this note, we prove \Cref{lem:finding}, restated below.

\Finding*

Our algorithm essentially generalizes Section~3 of~\cite{li2026balancing} to capacitated graphs. For completeness and consistency, we repeat the entire section almost verbatim, making changes only when necessary.

We interpret the balance of a flow as the \emph{minimum ratio cut} of the underlying \emph{undirected} graph, and then use the dynamic minimum ratio cut data structure from a recent interior-point method-based flow algorithm~\cite{van2024almost}.

We refer to edges of an undirected graph as \emph{edges} instead of arcs, and for a cut $(S,V\setminus S)$ in an undirected or directed graph, the undirected boundary $\partial S$ is the set of edges/arcs with exactly one endpoint in $S$. We now define the minimum ratio cut of an undirected graph, whose input includes a \emph{gradient function} $\bg:V\to\mathbb R$ with $\bg\perp\mathbbm 1$, i.e., $\sum_{v\in V}\bg(v)=0$.

\begin{definition}
The \emph{minimum ratio cut} of a positively weighted, undirected graph with gradient function $\bg\perp\mathbbm 1$ is the minimum ratio $\bg(S)/w(\partial S)$ over all $S\subseteq V$ with $w(\partial S)\ne0$.
\end{definition}

Observe that the requirement $\bg\perp\mathbbm 1$ means that $\bg(S)+\bg(V\setminus S)=0$ for any $S\subseteq V$, so either $S$ or $V\setminus S$ has non-positive ratio. Intuitively, the goal is to maximize $|\bg(S)|$ while minimizing $w(\partial S)$, and then take either $S$ or $V\setminus S$, whichever has non-positive ratio.

We construct our minimum ratio cut instance as follows. Given a weighted directed graph $G=(V,E,w)$, consider the underlying undirected graph together with the gradient function $\bg:V\to\mathbb R$ defined as
\[ \bg=\sum_{(u,v)\in E}w(u,v)\cdot(\mathbbm 1_u-\mathbbm 1_v) ,\]
i.e., for each arc $(u,v)$ in the directed graph, add $w(u,v)$ to the tail $u$ and subtract $w(u,v)$ from the head $v$. In other words, we ``undirect'' the graph and encode the arc directions through a gradient function.

\begin{lemma}
For any set $S\subseteq V$, we have $\bg(S)=w(\partial^+S)-w(\partial^-S)$.
\end{lemma}
\begin{proof}
Any arc $(u,v)\in\partial^+S$ contributes $w(u,v)$ to $\bg(S)$ since we add $w(u,v)$ to the tail $u\in S$, while the subtraction at the head $v\notin S$ is not counted in $\bg(S)$. Similarly, any arc $(u,v)\in\partial^-S$ contributes $-w(u,v)$ to $\bg(S)$. All other arcs contribute $0$ since either $u,v\notin S$ in which case there is no contribution, or $u,v\in S$ in which case the contributions of $+w(u,v)$ and $-w(u,v)$ cancel out.
\end{proof}

\begin{lemma}\label{lem:approx-balanced}
If the minimum ratio cut has ratio $\bg(S)/w(\partial S)\ge-\epsilon$, then $w(\partial^+S)\ge\frac{1-\epsilon}{1+\epsilon}w(\partial^-S)$ for all $S\subseteq V$.
\end{lemma}
\begin{proof}
It suffices to show that a cut $S\subseteq V$ with $w(\partial^+S)<\frac{1-\epsilon}{1+\epsilon}w(\partial^-S)$ implies that $\bg(S)/w(\partial S)<-\epsilon$. By the previous lemma, we have $\bg(S)=w(\partial^+S)-w(\partial^-S)$. Together with $w(\partial S)=w(\partial^+S)+w(\partial^-S)$, we obtain
\begin{align*}
\bg(S)+\epsilon\cdot w(\partial S)&=(w(\partial^+S)-w(\partial^-S))+\epsilon\cdot(w(\partial^+S)+w(\partial^-S))
\\&=(1+\epsilon)\cdot w(\partial^+S)-(1-\epsilon)\cdot w(\partial^-S)
\\&<(1-\epsilon)\cdot w(\partial^-S)-(1-\epsilon)\cdot w(\partial^-S)
\\&=0,
\end{align*}
or equivalently $\bg(S)/w(\partial S)<-\epsilon$, as desired.
\end{proof}

\subsection{Dynamic data structure}

We use the following data structure defined in~\cite{van2024almost}:
\begin{definition}[$\alpha$-approximate min-ratio cut data structure, Definition 5.2 of~\cite{van2024almost}]
For a weighted undirected graph $H=(V,E,\bu)$ where $\bu\in\mathbb R^E$ and $\bu(e)\in[1/U,U]$\footnote{The only difference is that~\cite{van2024almost} enforces $\bu(e)\in[1,U]$ while we require the weaker $\bu(e)\in[1/U,U]$. However, we can scale all $\bu(e)$ and $\bg(v)$ by $U$, giving $\bu(e)\in[1,U^2]$ and $\bg(v)\in[-U^2,U^2]$, which preserves all ratios and still satisfies $\log U^2=\tilde{O}(1)$. We also scale the initial $\by$ by $1/U$ and apply their data structure with each $\eta$ in \textsc{ToggleCut}$(\eta)$ scaled by $1/U$, so that the value $\bu(e)(\by(u)-\by(v))$ remains unchanged. Assuming the original constraint $\bu(e)\in[1,U]$ and tracking the scaling by $U$ is cumbersome, so we instead assume the weaker $\bu(e)\in[1/U,U]$ directly.} for $\log U=\tilde{O}(1)$, an initial gradient $\bg\in\mathbb R^V$ where $\bg\bot\mathbbm 1$ and $\bg(v)\in[-U,U]$, an initial potential vector $\by\in\mathbb R^V$, and a detection threshold parameter $\Delta$, an $\alpha$-approximate min-ratio cut data structure $\mathcal D$ implicitly maintains the potential vector $\by$ and supports the following operations:
 \begin{itemize}
 \item $\textsc{InsertEdge}(e), \textsc{DeleteEdge}(e)$: inserts/deletes the edge $e$ with weight $\bu(e)$.
 \item $\textsc{UpdateGradient}(u,v,\delta)$: updates $\bg(u)\gets\bg(u)+\delta$ and $\bg(v)\gets\bg(v)-\delta$.
 \item $\textsc{Potential}(v)$: returns $\by(v)$.
 \end{itemize}

The data structure maintains a tuple $(g,u)$ where $g\le0$ and $u>0$ such that for some implicit cut $C\subseteq V$ we have $\bg(C)=g$ and $\bu(\partial C)\le u$, and the (negative) ratio $g/u$ is at most $1/\alpha$ times the minimum possible over all $C\subseteq V$.

The data structure has an additional operation \textsc{ToggleCut}$(\eta)$ based on the current tuple $(g,u)$ with implicit cut $C$: given a parameter $0<\eta\le1/u$, the data structure implicitly updates $\by(v)\gets\by(v)+\eta$ for all $v\in C$. Then, the data structure returns some edge set $E'$ that includes every edge $e=\{u,v\}\in E$ for which the value $\bu(e)(\by(u)-\by(v))$ has changed by at least an additive $\Delta$ since it was inserted/last returned in $E'$.
\end{definition}

\begin{theorem}[Theorem 5.4 of~\cite{van2024almost}]\label{thm:min-ratio-cut-data-structure}
There is a deterministic $\alpha$-approximate min-ratio cut data structure for $\alpha=n^{o(1)}$ such that every operation is processed in amortized time $n^{o(1)}\log U$. Furthermore, the total number of edges returned by the algorithm after $t$ calls to \textsc{ToggleCut}$(\cdot)$ is at most $n^{o(1)}\cdot t/\Delta$.
\end{theorem}

\subsection{The algorithm}

The algorithm initializes an $\alpha$-approximate min-ratio cut data structure (\Cref{thm:min-ratio-cut-data-structure}) on the input graph, where $\by=0$ and all edges have weight $f(u,v)=\frac1{\max\{\by(v)-\by(u),0\}+1/c(u,v)}=c(u,v)$. The detection threshold parameter $\Delta$ is fixed to $\Delta=\epsilon/(8\alpha)$. The algorithm only updates arc/edge weights when they pass the detection threshold from the data structure. That is, whenever the data structure returns some edge set $E'$, the algorithm makes the following updates for each edge $e\in E'$ with corresponding arc $(u,v)$: first, call $\textsc{DeleteEdge}(e)$ and $\textsc{UpdateGradient}(u,v,-\bu(e))$ to remove the old contribution, and then call $\textsc{InsertEdge}(u,v,f(u,v))$ and $\textsc{UpdateGradient}(u,v,f(u,v))$, where the new weight $f(u,v)=\frac1{\max\{\by(v)-\by(u),0\}+1/c(u,v)}$ is computed by querying $\textsc{Potential}(u)$ and $\textsc{Potential}(v)$ to determine $\by(u)$ and $\by(v)$.

While the data structure maintains an implicit cut $C\subseteq V$ with $g/u\le-\epsilon/\alpha$, the algorithm calls $\textsc{ToggleCut}(\eta)$ with $\eta=\frac\epsilon{16\alpha u}$. Otherwise, if $g/u>-\epsilon/\alpha$, then since the data structure is $\alpha$-approximate, the minimum ratio cut is at least $-\epsilon$. By \Cref{lem:approx-balanced}, we have $w(\partial^+S)\ge\frac{1-\epsilon}{1+\epsilon}w(\partial^-S)$ for all $S\subseteq V$.

\subsection{Stability}

Since edge weights only get updated when they are returned in $E'$, the weights $w(u,v)$ maintained by the data structure only approximate the true flow values $f(u,v)$ for the current potential function $\by$. We begin with a stability guarantee on the weight function that is used throughout the analysis.
\begin{lemma}[Stability]\label{lem:stability}
Consider three values $x_1,x_2\in\mathbb R$ and $c>0$, and define the weights\linebreak $w_i=\frac1{\max\{x_i,0\}+1/c}$ for $i\in\{1,2\}$. For any parameter $0\le\delta<1$, if $|x_1-x_2|\le\delta/w_1$, then $\frac1{1+\delta}w_1\le w_2\le\frac1{1-\delta}w_1$.
\end{lemma}
\begin{proof}
We bound
\begin{align*}
w_2=\frac1{\max\{x_2,0\}+1/c}
\ge\frac1{\max\{x_1+\delta/w_1,0\}+1/c}
\ge\frac1{\max\{x_1,0\}+1/c+\delta/w_1}
&=\frac1{1/w_1+\delta/w_1}
\\&=\frac{w_1}{1+\delta}
\end{align*}
and
\begin{align*}
w_2=\frac1{\max\{x_2,0\}+1/c}
\le\frac1{\max\{x_1-\delta/w_1,0\}+1/c}
\le\frac1{\max\{x_1,0\}+1/c-\delta/w_1}
&=\frac1{1/w_1-\delta/w_1}
\\&=\frac{w_1}{1-\delta}
\end{align*}
as desired.
\end{proof}

\begin{lemma}[Approximation of weights]\label{lem:approx-weights}
For any arc $(u,v)$, we have $\frac1{1+\Delta} w(u,v)\le f(u,v)\le\frac1{1-\Delta} w(u,v)$.
\end{lemma}
\begin{proof}
Let $\tilde\by$ be the potential function when the edge $(u,v)$ was last inserted or returned in $E'$. In particular, $w(u,v)=\frac1{\max\{\tilde\by(v)-\tilde\by(u),0\}+1/c(u,v)}$.
By the guarantee of the data structure, we have
\[ w(u,v)\cdot|(\tilde\by(u)-\tilde\by(v))-(\by(u)-\by(v))|<\Delta \iff |(\tilde\by(u)-\tilde\by(v))-(\by(u)-\by(v))|<\Delta/ w(u,v) .\]
Applying \Cref{lem:stability} with $x_1=\tilde\by(v)-\tilde\by(u)$, $x_2=\by(v)-\by(u)$, $c=c(u,v)$, and $\delta=\Delta$ finishes the proof.
\end{proof}

Finally, we show that if the algorithm stops calling $\textsc{ToggleCut}(\eta)$, at which point $g/u>-\epsilon/\alpha$, then $f$ is a $(1+O(\epsilon))$-pseudo-circulation. To obtain the $(1+\epsilon)$-pseudo-circulation promised by \Cref{lem:finding}, simply run the algorithm with $\epsilon$ scaled down by a constant factor.
\begin{lemma}
If $g/u>-\epsilon/\alpha$, then $f(\partial^+S)\ge(1-O(\epsilon))f(\partial^-S)$ for all $S\subseteq V$.
\end{lemma}
\begin{proof}
Since the data structure is $\alpha$-approximate, the minimum ratio cut is at least $-\epsilon$. By \Cref{lem:approx-balanced}, we have $w(\partial^+S)\ge\frac{1-\epsilon}{1+\epsilon}w(\partial^-S)$ for all $S\subseteq V$. Together with \Cref{lem:approx-weights}, we have $(1+\Delta)f(\partial^+S)\ge w(\partial^+S)\ge\frac{1-\epsilon}{1+\epsilon}w(\partial^-S)\ge\frac{1-\epsilon}{1+\epsilon}\cdot(1-\Delta)f(\partial^-S)$, so $(1+O(\epsilon))f(\partial^+S)\ge f(\partial^-S)$, as promised.
\end{proof}

\subsection{Energy bounds}\label{sec:energy-bounds}

To bound the running time, we need to track the number of calls to \textsc{ToggleCut}. Let $M$ be the maximum value of $\max_{u,v}|\by(u)-\by(v)|$ over the course of the algorithm, and define the \emph{energy} of arc $(u,v)$ as
\[ \mathcal E(u,v)=\int_{\by(v)-\by(u)}^M\frac1{\max\{x,0\}+1/c(u,v)}\,dx .\]
Observe that
\[ \mathcal E(u,v)=
\begin{cases}
\ln(M+1/c(u,v))-\ln(\by(v)-\by(u)+1/c(u,v))&\text{if }\by(v)\ge\by(u),
\\\ln(M+1/c(u,v))-\ln(1/c(u,v))+c(u,v)\cdot(\by(u)-\by(v))&\text{if }\by(v)\le\by(u).\end{cases} \]

The following lemma shows that each call to \textsc{ToggleCut} decreases the total energy by an additive $1/n^{o(1)}$.
\begin{lemma}[Energy decrease]\label{lem:energy-decrease}
On each call to \textsc{ToggleCut}, the total energy decreases by at least $\Omega(\epsilon^2/\alpha^3)$.
\end{lemma}
\begin{proof}
On each call to $\textsc{ToggleCut}(\eta)$, the implicit cut $C$ satisfies $g/u\le-\epsilon/\alpha$, where $\bg(C)=g<0$ and $\bu(\partial C)\le u$. We also have $\bu(\partial C)\ge u/\alpha$ since otherwise, if $\bu(\partial C)<u/\alpha$, then the ratio $\bg(C)/\bu(\partial C)$ would be less than $\bg(C)/(u/\alpha)=\alpha\cdot g/u$, contradicting the fact that $g/u$ is at most $1/\alpha$ times the min-ratio cut.

We now analyze the contribution of $\textsc{ToggleCut}(\eta)$ to the total energy. Let $w$ and $\by$ be the weight and potential function immediately before $\textsc{ToggleCut}(\eta)$. For each arc $(u,v)\in\partial^+C$, the value of $\by(u)$ is increased by $\eta$, so the energy $\mathcal E(u,v)$ increases by
\[ \int_{\by(v)-\by(u)-\eta}^{\by(v)-\by(u)}\frac1{\max\{x,0\}+1/c(u,v)}\,dx .\]
Similarly, for each arc $(u,v)\in\partial^-C$, the value of $\by(v)$ is increased by $\eta$, so the energy $\mathcal E(u,v)$ decreases by
\[ \int_{\by(v)-\by(u)}^{\by(v)-\by(u)+\eta}\frac1{\max\{x,0\}+1/c(u,v)}\,dx .\]
By the approximation of weights (\Cref{lem:approx-weights}), we have $f(u,v)\le2w(u,v)$, so for any arc $(u,v)\in\partial C$, we have $u\ge\bu(\partial C)\ge w(u,v)\ge\frac12f(u,v)$.
Since $\eta=\frac\epsilon{16\alpha u}\le\frac\epsilon{8\alpha\cdot f(u,v)}$, by the stability guarantee (\Cref{lem:stability}) with $\delta=\frac\epsilon{8\alpha}$, we have
\begin{align*}
\frac1{1+\epsilon/(8\alpha)}f(u,v)\le\frac1{\max\{x,0\}+1/c(u,v)}&\le\frac1{1-\epsilon/(8\alpha)}f(u,v)
\\&\qquad\qquad\text{ for all }x\in[\by(v)-\by(u)-\eta,\,\by(v)-\by(u)+\eta].
\end{align*}
In particular, the energy increase of each arc $(u,v)\in\partial^+C$ is at most $\eta\cdot\frac1{1-\epsilon/(8\alpha)}\cdot f(u,v)$, and the energy decrease from each arc $(u,v)\in\partial^-C$ is at least $\eta\cdot\frac1{1+\epsilon/(8\alpha)}\cdot f(u,v)$. Summed over all such arcs, the net change in energy is at most
\[ \eta\cdot\frac1{1-\epsilon/(8\alpha)}\cdot f(\partial^+C)-\eta\cdot\frac1{1+\epsilon/(8\alpha)}\cdot f(\partial^-C) .\]

Recall that $\bg(C)/\bu(\partial C)\le g/u\le-\epsilon/\alpha$ and $\bg(C)=w(\partial^+C)-w(\partial^-C)$ and $\bu(\partial C)=w(\partial^+C)+w(\partial^-C)$. We have
\begin{align*}
\alpha\cdot\bg(C)+\epsilon\cdot\bu(\partial C)\le0&\iff(\alpha+\epsilon)\cdot w(\partial^+C)-(\alpha-\epsilon)\cdot w(\partial^-C)\le0
\\&\iff w(\partial^+C)\le\frac{\alpha-\epsilon}{\alpha+\epsilon}\cdot w(\partial^-C).
\end{align*}
Together with the approximation of weights (\Cref{lem:approx-weights}), we have
\[ f(\partial^+C)\le\frac1{1-\Delta}\cdot w(\partial^+C)\le\frac1{1-\Delta}\cdot\frac{\alpha-\epsilon}{\alpha+\epsilon}\cdot w(\partial^-C)\le\frac1{1-\Delta}\cdot\frac{\alpha-\epsilon}{\alpha+\epsilon}\cdot(1+\Delta)\cdot f(\partial^-C). \]
Since $\Delta=\epsilon/(8\alpha)$, it follows that the net change in energy is at most
\begin{align*}
&\eta\cdot\frac1{1-\epsilon/(8\alpha)}\cdot f(\partial^+C)-\eta\cdot\frac1{1+\epsilon/(8\alpha)}\cdot f(\partial^-C)
\\\le{}&\eta\cdot\frac1{1-\epsilon/(8\alpha)}\cdot\frac{1+\epsilon/(8\alpha)}{1-\epsilon/(8\alpha)}\cdot\frac{\alpha-\epsilon}{\alpha+\epsilon}\cdot f(\partial^-C)-\eta\cdot\frac1{1+\epsilon/(8\alpha)}\cdot f(\partial^-C)
\\\le{}&\eta\cdot-\Omega\left(\frac\epsilon\alpha\right)\cdot f(\partial^-C)
\\\le{}&\eta\cdot-\Omega\left(\frac\epsilon\alpha\right)\cdot w(\partial^-C).
\end{align*}
Recall that $\bu(\partial C)\ge u/\alpha$, which means that $\eta=\frac\epsilon{16\alpha u}\ge\frac\epsilon{16\alpha^2\bu(\partial C)}$. Also, since $\bg(C)<0$, we have $w(\partial^-C)\ge\frac12w(\partial C)=\frac12\bu(\partial C)$. We conclude that the net change in energy is at most
\[ \frac\epsilon{16\alpha^2\bu(\partial C)}\cdot-\Omega\left(\frac\epsilon\alpha\right)\cdot\frac12\bu(\partial C)=-\Omega\left(\frac{\epsilon^2}{\alpha^3}\right), \]
as desired.
\end{proof}

\subsection{Running time}

By \Cref{thm:min-ratio-cut-data-structure}, the algorithm runs in $n^{o(1)}\cdot t/\Delta$ where $t$ is the number of calls to \textsc{ToggleCut}. Initially, each arc has energy at most $\ln(M+1)+\ln U$, so the total energy is at most $m\cdot(\ln(M+1)+\ln U)$. By \Cref{lem:energy-decrease}, each call to \textsc{ToggleCut} decreases the total energy by $\Omega(\epsilon^2/\alpha^3)=\epsilon^2/n^{o(1)}$, so $t\le m^{1+o(1)}(\ln(M+1)+\ln U)/\epsilon^2$ and the total running time is $m^{1+o(1)}(\ln(M+1)+\ln U)/\epsilon^3$.

Since arc capacities are polynomially bounded, we have $U=n^{O(1)}$. Finally, we claim that $M=n^{O(1)}$ as well. Suppose for contradiction that $M$ is too large, and consider the moment of the algorithm at which $\by(u)-\by(v)=M$ for some $u,v\in V$. Consider a path from $u$ to $v$ of at most $n-1$ arcs in the strongly connected graph. There is an arc $(u',v')$ on the path with $\by(u')-\by(v')\ge M/(n-1)$, which has energy at least $c(u',v')\cdot(\by(u')-\by(v'))\ge M/(n-1)$. In particular, the total energy is at least $M/(n-1)$. However, by the previous argument, the total energy starts at $m\cdot(\ln(M+1)+\ln U)$ and can only decrease. For large enough $M=n^{O(1)}$, we obtain a contradiction.

\subsection*{Acknowledgement} The authors used ChatGPT 5.6 Sol Pro to assist with calculations and review. The authors assume responsibility for all content.

\bibliographystyle{alpha}
\bibliography{ref}

\end{document}